\documentclass[journal,letterpaper]{IEEEtran}

\usepackage[dvips]{graphicx}
\usepackage{amsmath,amsfonts,amssymb,amsthm,mathrsfs}
\usepackage{bbold}

\newtheorem{thm}{Theorem}
\newtheorem{lem}{Lemma}
\newtheorem{defi}{Definition}

\theoremstyle{remark}
\newtheorem*{rem}{Remark}
\newtheorem{cor}{Corollary}

\newcommand{\ket}[1]{\mathop{\left|#1\right>}\nolimits}       

\newcommand{\kb}[2]{| #1\rangle\!\langle #2 |}
\newcommand{\Tr}[1]{\mathop{{\mathrm{Tr}}_{#1}}}              
\newcommand{\Det}{\mathop{{\mathrm{Det}}}}

\newcommand{\bbC}{{\mathbb C}}
\def\hsp{\hspace{.567cm}}
\newcommand{\dg}{\dagger}
\newcommand{\smfrac}[2]{\mbox{$\frac{#1}{#2}$}}
\newcommand{\Smin}[1]{\mathop{{\mathit{S^{\rm min}}{\left(#1\right)}}}}   
\newcommand{\Sminp}[1]{\mathop{{\mathit{S_p^{\rm min}}{\left(#1\right)}}}}   
\def\openone{\mathbb{1}}
\newcommand{\vac}{\mathsf{vac}}
\def\a{\alpha}
\def\b{\beta}

\def\N{\mathcal{N}}
\def\M{\mathcal{M}}
\def\D{\mathcal{D}}
\newcommand{\Cl}[1]{\mathcal{C}l_{1,#1}}               

\def\Dcd{{\check{\D}}}
\def\Dthree{{\D_3}}
\def\Dell{{\D_\ell}}

\def\Dcdell{{\!\!\check{\,\,\D_\ell}}}
\def\C{\mathcal{C}}
\def\A{\mathcal{A}}
\def\G{\mathcal{G}}

\def\C{\mathcal{C}}

\def\P{\mathcal{P}}
\def\S{\mathcal{S}}
\def\B{\mathcal{B}}

\def\U{\mathcal{U}}
\def\S{\mathcal{S}}
\def\T{\mathcal{T}}

\def\HH{\mathscr{H}}       
\def\FF{\mathscr{F}}

\def\lam{\lambda}
\def\e{\varepsilon}
\def\k{\kappa}
\def\o{\omega}
\def\r{\varrho}
\def\s{\sigma}

\def\p{\prime}

\begin{document}
\title{An Infinite Sequence of Additive Channels: \\the Classical Capacity of Cloning Channels}

\author{Kamil Br\'adler
\thanks{School of Computer Science, McGill University, Montreal, Canada}}

\maketitle

\begin{abstract}
    We introduce an infinite sequence of quantum channels for which the Holevo capacity is additive. The channel series is closely related to the quantum channels arising from universal quantum cloning machines. The additivity proof is motivated by a special property the studied channels enjoy:  the property of conjugate degradability. As a consequence of the announced proof, we also provide an easy way of proving the additivity of the Holevo capacity for the original Unruh channel for which the quantum capacity is already known. Consequently, we present not only an infinite series of finite-dimensional channels but also a nontrivial example of an infinite-dimensional channel for which the classical and quantum channel capacities are easily calculable.
\end{abstract}


\section{Introduction}
Recently, a notorious open problem in quantum information theory known as the additivity of the Holevo capacity was finally resolved~\cite{counterexSmin1} with the negative answer. The article culminated a long period of waiting for the answer to the question (later a conjecture) which appeared shortly after people started to ask about the role of quantum correlations for information theory~\cite{HolevoCapAdd}. The former conjecture states that entangled states do not improve the classical capacity of quantum channels. Quantum channel $\N$ is a completely positive (CP) map $\N:\FF\big(\HH^{(I)}_{\rm in}\big)\to\FF\big(\HH^{(O)}_{\rm out}\big)$. $\FF\big(\HH^{(K)}\big)$ is the state space occupied by Hermitean  trace one  operators and  $\HH^{(K)}$ denotes a $K-$dimensional Hilbert space. The ultimate formula for the classical capacity is $C=\lim_{n\to\infty}{{1\over n}C_{\rm Hol}\left(\N^{\otimes n}\right)}$. $C_{\rm Hol}(\N)$ is the Holevo capacity of a channel~$\N$~\cite{HSW} defined as
\begin{equation}\label{HolevoCap}
    C_{\rm Hol}(\N)\stackrel{\rm df}{=}\sup_{\{p_i\r_i\}}{\bigg\{S\bigg(\sum_ip_i\N(\r_i)\bigg)-\sum_i p_iS\bigg(\N(\r_i)\bigg)\bigg\}},
\end{equation}
where $\{p_i\r_i\}$ is the input ensemble $\r=\sum_ip_i\r_i$ and $S(\s)=-\Tr{}\s\log\s$ is the von Neumann entropy~\footnote{$\log$ gives the logarithm to base two.}.
The calculation of $C$ appears to be an intractable problem. The conjecture claimed that $C_{\rm Hol}(\N_1\otimes\N_2)=C_{\rm Hol}(\N_1)+C_{\rm Hol}(\N_2)$ for arbitrary channels $\N_1,\N_2$. This condition is slightly stronger (strong additivity) than if $\N_1=\N_2$ (weak additivity). One can immediately see how the calculation of $C$ might have been much simpler if the conjecture had been correct. Let us stress, however, that even if the conjecture does not hold in general there are important classes of channels for which it holds~\cite{DepolCh,TransDepol}.

The final disproof of the conjecture would not be possible without many important intermediate results. First, it was shown that the additivity of the Holevo capacity is globally (that is, not for a particular channel) equivalent to other additivity questions~\cite{ConjectsEquiv}, particularly to the additivity of the minimum output entropy (MOE)~\cite{Smin}. The MOE belongs to the more general class of entropies known as the minimum output R\'enyi entropy (MORE). The MORE of a channel $\N$ is defined
\begin{equation}\label{Smin}
    \Sminp{\N}\stackrel{\rm df}{=}\min_{\r}{\{S_p(\N(\r))\}},\hsp\r\in\FF(\HH_{\rm in}),
\end{equation}
where $S_p(\r)=(1-p)^{-1}\log{\Tr{}\r^p}$ is the R\'enyi entropy (for $p\to1^+$ we get the von Neumann entropy). The MORE conjecture was disproved for various intervals of $p$ (for $p>1$ in~\cite{counterexsSminp_g1} and for $p\to0$ in~\cite{counterexsSminp_0}) and, as indicated, at last also for $p=1$~\cite{counterexSmin1}. Note that by the concavity of entropy we may restrict ourselves to the minimization over input pure states.

The motivation for the classical capacity study presented in this paper
comes from the analysis of another (infinite-dimensional) quantum channel which appears in the context of quantum field theory in curved spacetime -- the Unruh channel~\cite{UnruhCh}. The Unruh channel occupies an important place in the field of relativistic quantum information and quantum field theory due to its close relationship to two fundamental processes: (i) black hole evaporation or more generally to the process of black hole stimulated emission and (ii) transformation of states prepared by an inertial observer as seen by a uniformly accelerating observer (the Unruh effect~\cite{Unruh76}). We observed~\cite{UnruhCh} that the Unruh channel has a direct sum structure and therefore decomposes into an infinite sequence of finite-dimensional channels which are closely related to the channels arising from universal quantum cloning machines (UQCM) for qubits~\cite{cloners}. We will call them cloning channels and we will prove that there exists a single-letter formula for the Holevo capacity for all of them. It is also known that cloning channels are conjugate degradable~\cite{ConjDeg}.  Channels are called conjugate degradable by virtue of existence of a conjugate degrading map transforming the output of the channel to its complementary output up to complex conjugation. It has been recently shown that the optimized coherent information of conjugate degradable channels is additive~\cite{ConjDeg}. Therefore, as a result of this paper we obtain an infinite sequence of channels for which both the classical and quantum capacity can be calculated in an easy way. As an aside we will show that all studied cloning channels are not only conjugate-degradable but also degradable. Originally, this seemed to be a difficult task to directly~\cite{cloners} prove it.  Here we found a relatively straightforward way to show this fact and conjugate degradability of cloning channels again played an important role in the proof. The second main result of this paper is the proof of additivity of the Holevo capacity for the Unruh channel itself. It provides us with a non-trivial example of an infinite-dimensional channel for which both the classical and quantum capacity are known and easily calculable.

In section~\ref{sec:Unruhchannel} we briefly recall the properties of the Unruh channel~\cite{UnruhCh} and present its decomposition into a sequence of finite-dimensional channels.  Section~\ref{sec:infseq} contains the main result of the paper. Using the structure of cloning channels we prove in the first part that (i) the Holevo capacity of cloning channels is additive and (ii)  all cloning channels are degradable by showing that their complementary channels are entanglement-breaking. In the second part of section~\ref{sec:infseq} we prove the additivity of the Holevo capacity for the infinite-dimensional Unruh channel. We conclude the paper with a technical tool to actually determine the form of degrading channels for cloning channels and illustrate it on a few examples.

If not stated otherwise, note that in this paper by additivity of a channel we mean strong additivity of the Holevo capacity.

\section{Unruh channel}\label{sec:Unruhchannel}
In this section we briefly review the definition and properties of the Unruh channel~\cite{UnruhCh}. The channel naturally appears as the transformation of a photonic qubit prepared by a stationary Minkowski observer if it is detected by a uniformly accelerated observer. It is well known that  inertial and non-inertial observers cannot agree on the notion of a particle. The most dramatic example is the Minkowski vacuum seen by an non-inertial observer as a thermally populated state~\cite{Unruh76}.

In the same spirit, a pure qubit prepared in the Hilbert space of a Minkowski observer  is seen as an infinite-dimensional mixed state in the Hilbert space of the accelerated observer. The responsible transformation reads
\begin{multline}\label{squeeze}
       U_{AC}(r)=\smfrac{1}{\cosh^2{r}}e^{\tanh{r}(a_1^\dg c_1^\dg+a_2^\dg c_2^\dg)}\\
        \times e^{-\ln{\cosh{r}}(a_1^\dg a_1+c_1^\dg c_1+a_2^\dg a_2+c_2^\dg c_2)}e^{-\tanh{r}(a_1c_1+a_2c_2)},
\end{multline}
where $r$ is the proper acceleration of the inertial observer and $AC$ denotes the output Hilbert space as a shorthand for the tensor product of two Hilbert spaces $A$ and $C$. Note that the creation and annihilation operators satisfy the so-called canonical commutation relations
\begin{subequations}
\begin{align}\label{eq:CCR}
[a_i,a^\dg_{j}]&=\delta(i-j),\\
[a^\dg_i,a^\dg_{j}]&=[a_i,a_{j}]=0,
\end{align}
\end{subequations}
where the square brackets denote the commutator and all the degrees of freedom (discrete or continuous) are labeled by $i,j$ (similarly for $c_i$). The exponentiation in Eq.~(\ref{squeeze}) is understood in the operator sense. For an input state $\ket{\psi}_{A'C'} = (\beta a_2^\dg + \alpha a_1^\dg)\ket{\vac}_{A'C'}$, we can further simplify $\ket{\phi}_{AC}=U_{AC}(r)\ket{\psi}_{A'C'}$ as
\begin{align}\label{simp_squeeze}
     \ket{\phi}_{AC}={1\over\cosh^3{r}}&(\beta a_2^\dg+\alpha a_1^\dg)\nonumber\\
     &\times\exp{[\tanh{r}(a_1^\dg c_1^\dg+a_2^\dg c_2^\dg)]}\ket{\vac}_{A'C'}.
\end{align}
From a physical point of view, the modes $c_i$ appear beyond the event horizon of the accelerated observer and are therefore unobservable. Tracing over them, we get a state with an interesting structure, further investigated in~\cite{UnruhCh}. If we reorder the basis according to the total number of incoming photons in modes $a_i$,  we obtain an infinite-dimensional block-diagonal density matrix
\begin{equation}\label{blockdiag}
    \sigma_A=1/2(1-z)^3\bigoplus_{\ell=2}^\infty \ell(\ell-1)z^{\ell-2}\e_\ell,
\end{equation}
where $0\leq z<1,z=\tanh^2{r}$. The states $\e_\ell$ and the corresponding quantum channel will be studied in the next section.

The transformation leading to Eq.~(\ref{blockdiag}) has already been  studied before in a context mentioned in the introduction. The authors of Ref.~\cite{blackhole_clone} analyzed the process of black hole stimulated emission induced by impinging photonic qubits. The stimulated emission dynamics is governed by exactly the same Hamiltonian as the one leading to the unitary operator $U_{AC}$. The reason for this formal similarity lies in the linear relations known as Bogoliubov transformation~\cite{bogo}. Bogoliubov transformation connects the creation and annihilation operators  of the Hilbert space of a Minkowski observer and a uniformly accelerating observer in our case and similarly the Hilbert space of a freely falling observer and an observer in a distant future in case of Ref.~\cite{blackhole_clone}. In the former case the physical parameter of the evolution operator is the proper acceleration $r$ and in the latter case it is the black hole surface gravity~\footnote{Note that the particle-antiparticle basis used in~\cite{blackhole_clone}  exactly corresponds to the dual-rail encoding in which an input state $\ket{\psi}$ is written.}. Even more interestingly, as observed in~\cite{blackhole_clone}, the same Hamiltonian is closely related to the $N\to M$ universal cloning machine for qubits~\cite{cloners} ($\ell=M+1$). In other words, if an observer throws an $N-$qubit photonic state into a black hole (the state is already symmetrized due to the bosonic nature of the photons) another observer in a distant future gets $M$ approximate copies depending on the total  number $M$ of photons he measures. We will study the explicit output of  Eq.~(\ref{blockdiag}) which corresponds to  the case of $1\to(\ell-1)$ cloning machines. 

\section{Additivity of the classical capacity}\label{sec:infseq}
\subsection{The classical  capacity of cloning channels}\label{subsec:infseq}

The previous section served as a physical motivation for the appearance of UQCMs for qubits. In this section we observe that the cloning channels `constitute' the corresponding Unruh  channel in a very specific way. Namely, we will show that the additivity of the Holevo capacity for the $1\to2$ cloning channel implies the additivity of the Holevo capacity for all $1\to(\ell-1)$ cloning channels (that is for all $\ell>3$). Another consequence will be the proof of additivity of the Holevo capacity for the Unruh channel itself.

We first recall the definition of {\it unitarily covariant channels}  introduced in~\cite{CovCh}.
\begin{defi}\label{def:covariance}
Let $G$ be a unitary compact group of Lie type and let  $r_1(g)\in\HH_{\rm in},r_2(g)\in\HH_{\rm out}$  be irreps of $g\in G$. A channel $\N:\FF\big(\HH_{\rm in}\big)\to\FF\big(\HH_{\rm out}\big)$ is unitarily covariant if
\begin{equation}\label{covcond}
    \N\left(r_1(g)\r r_1(g)^\dg\right)=r_2(g)\N(\r)r_2(g)^\dg
\end{equation}
holds for all $\r$ and $g \in G$.
\end{defi}
In the following text, by covariant we mean unitarily covariant. It has been shown that for any covariant channel the following equivalence condition holds
\begin{equation}\label{equivlocal}
    C_{\rm Hol}(\N)=\log{f}-\Smin{\N},
\end{equation}
where $f=\dim{\HH_{\rm out}}$. Nevertheless, for Eq.~(\ref{equivlocal}) to hold the conditions in Definition~\ref{def:covariance} are not necessary and can be relaxed~\cite{WolfEisert}.

Let us stress that  we will leave the domain of the Fock space and adopt new notation. From now on, $\ket{n}$ represents a qudit living in an abstract Hilbert space~$\HH$ and not a Fock state of $n$ photons like in Section~\ref{sec:Unruhchannel}. The reason is that the Fock space formalism is a bit clumsy for the quantum information considerations which will follow. We will make occasional connections from one formalism to another to avoid possible confusion.

Let $W$ be the Hilbert space isometry $W:A'\hookrightarrow AC$ such that $W(\varphi_{A'})=U_{AC}^{(K)}(\ket{\varphi}_{A'}\ket{0}_{C'})$. $U_{AC}^{(K)}$ is a $K-$dimensional unitary transformation defined by its action on an input pure state $\ket{\varphi}_{A'}=\alpha\ket{0}+\beta\ket{1}$ and a reference state $\ket{0}$
\begin{multline}\label{unitary}
    \ket{\varphi}_{A'}\ket{0}_{C'}\xrightarrow{U_{AC}^{(K)}} {\sqrt{2\over(k+1)(k+2)}}\\
    \times\Biggl(\sum_{n=0}^k\alpha\sqrt{n+1}\ket{k-n}_A\ket{k-n}_C\\
    +\beta\sqrt{k-n+1}\ket{k-n+1}_A\ket{k-n}_C\Biggr),
\end{multline}
(thus $K=2(k+1)$). This unitary operation induces a class of CP maps $\e_\ell\stackrel{\rm df}{=}\Tr{C}\left[W(\varphi_{A'})\right]=\Cl{\ell-1}(\varphi_{A'})$ which we will call $1\to(\ell-1)$ cloning channels ($\ell=k+2$). We justify the name in a subsequent part of this section. The explicit output of $\Cl{\ell-1}(\varphi_{A'})$ and the corresponding complementary channel $\bar\k_{\ell-1}\stackrel{\rm df}{=}\S^c_{\ell-1}(\varphi_{A'})=\Tr{A}\left[W(\varphi_{A'})\right]$ for an input qubit $\varphi_{A'}=\openone/2+\vec{n}\cdot\vec{J}^{(2)}$ read
\begin{align}\label{StokesplusConjugateStokes1}
    \e_\ell&={2\over\ell(\ell-1)}\Big(\openone^{(\ell)}(\ell-1)/2+\sum_{i=x,y,z}n_iJ_i^{(\ell)}\Big)\\
    \label{StokesplusConjugateStokes2}
    \bar\k_{\ell-1}&={2\over\ell(\ell-1)}\Big(\openone^{(\ell-1)}\ell/2+\sum_{i=x,y,z}\tilde n_iJ_i^{(\ell-1)}\Big),
\end{align}
where $J_i^{(\ell)}$ are related to the $\ell-$dimensional generators of the $su(2)$ algebra. The $su(2)$ algebra generators are defined~\footnote{More precisely, the $su(2)$ algebra is a compact real form of the special linear algebra $sl(2,\bbC)$.} by
$\big[J^{(\ell)}_+,J^{(\ell)}_-\big]=2J^{(\ell)}_z,\big[J^{(\ell)}_z,J^{(\ell)}_{\pm}\big]=\pm J^{(\ell)}_\pm$ and in the above equations we use $J_x^{(\ell)}=1/2\big(J^{(\ell)}_++J^{(\ell)}_-\big),J^{(\ell)}_y=-i/2\big(J^{(\ell)}_+-J^{(\ell)}_-\big)$. We also defined $n_x=\a\bar \b+\bar \a\b,n_y=i(\a\bar \b-\bar \a\b),n_z=|\a|^2-|\b|^2$ and $\tilde n_x=n_x, \tilde n_y=-n_y,\tilde n_z=n_z$. For the purposes of this paper we consider only input pure states $\|\vec{n}\|_2=1$. Note that  barred operator $\bar\k_{\ell-1}$ indicates its entry-wise complex conjugation which results in transposition for density matrices.

States in Eq.~(\ref{StokesplusConjugateStokes1}) are the states $\varepsilon_\ell$ from Eq.~(\ref{blockdiag}) but stripped of all quantum-optical interpretations.  One could get the same states as outputs of the $1\to (\ell-1)$ qubit UQCM if we traced over the reference system and the channel output was rewritten in the completely symmetric (fixed) basis of $\ell-1$ qubits. Recall that a UQCM is defined as a unitary producing approximate copies of an unknown input state. The output states are in a well defined sense closest to the input state as much as it is allowed by the laws of quantum mechanics and under several reasonable assumptions (such as $SU(2)$ covariance and symmetry with respect to the permutations of the output clones~\cite{cloners_review}). Henceforth, we consider cloning channels $\Cl{\ell-1}$ to be CP maps whose output is composed of all $\ell-1$ clones. For further a generalization to the qudit case we invite readers to Refs.~\cite{capregion_qudits} and \cite{CMP}.

Comparing an input state $\varphi$ with an output $\varepsilon_\ell$ of $\Cl{\ell-1}$ we see that  the transformation preserves the Stokes parameters $n_i$, even as the dimension of the algebra representation changes. We may interpret $\Cl{\ell-1}$ as an input state representation-changing channel. Similarly, the complementary channels $\S^c_{\ell-1}$ also change the representation of the input state accompanied by transposition (complex conjugation). We will make use of this intriguing interpretation of these channels in the proof of Lemma~\ref{lem:S_ell-entbreaking}. Finally, we observe that $\Cl{\ell-1}$ is a covariant channel  (all UQCMs are by definition covariant) and so is $\S^c_{\ell-1}$.

Recall that if $\ell=2$, $\Cl{1}$ is an identity map and the complementary map $\S_1^c$ is just an ordinary trace map. Some interesting things start to happen for $\ell=3$ where $\S^c_{2}(\varphi_{A'})=1/3(\bar\varphi+\openone)$. This is an instance of the transpose depolarizing channel (alias the optimal transposition map for qubits) whose  Holevo capacity is known to be strongly additive~\cite{TransDepol}. It follows that its complement $\Cl{2}$ is strongly additive too~\cite{ComplCh}.

\begin{thm}\label{thm1}
   Cloning channels $\Cl{\ell-1}$ are additive for all~$\ell\geq2$.
\end{thm}
Before proving the theorem we first introduce the concept of conjugate degradability followed by a useful lemma. The definition of conjugate degradability~\cite{ConjDeg} resembles the one of degradability~\cite{DegCh} which we present for the sake of completeness.
\begin{defi}
(i) A channel $\N$ is degradable if there exists a map $\D$ called a degrading map which degrades the channel to its complementary channel $\N^c$
\begin{equation}\label{defCD}
    \D\circ\N=\N^c.
\end{equation}
We say that a channel is anti-degradable if its complementary channel is degradable.\\
(ii) A channel $\N$ is conjugate degradable if there exists a map $\Dcd$ called a conjugate degrading map which degrades the channel to its complementary channel $\N^c$ up to complex conjugation $\C$
\begin{equation}\label{defCD}
    \Dcd\circ\N=\C\circ\N^c.
\end{equation}
\end{defi}
A single-letter quantum capacity formula  exists for all degradable and conjugate degradable channels~\cite{DegCh,ConjDeg}.
\begin{lem}\label{lem:S_ell-entbreaking}
The complementary channels of all cloning channels $\Cl{\ell-1}$ are entanglement-breaking.
\end{lem}
\begin{figure}[t]
    \begin{center}
    \resizebox{5cm}{4cm}{\includegraphics{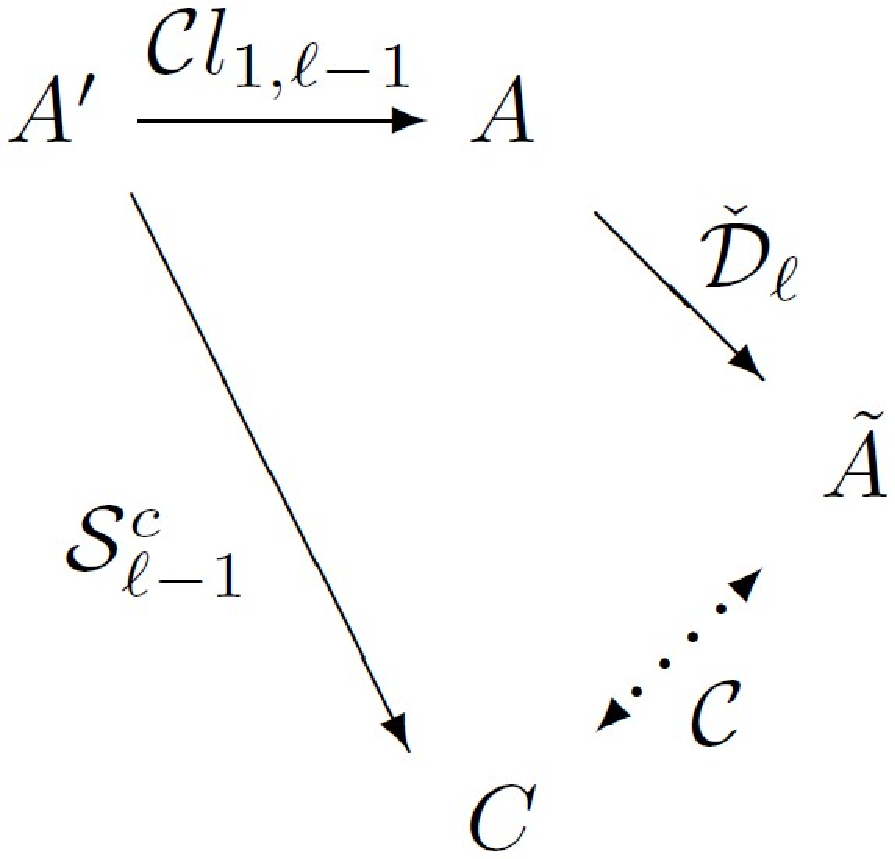}}
    \caption{\label{diag:chnl_struct}In the diagram, $\Cl{\ell-1}$ is a $1\to (\ell-1)$ cloning channel with a conjugate degrading map $\Dcdell$. The dotted line signalizing a non-CP map is complex conjugation $\C:\bar\k_{\ell-1}\leftrightarrow\k_{\ell-1}$.}
        \end{center}
\end{figure}
Looking at the diagram in Fig.~\ref{diag:chnl_struct} we recall an observation made in~\cite{ConjDeg}. The complementary channel of a conjugate degradable channel is either entanglement-breaking or entanglement-binding~\cite{entbind}. The reason is that the transposed output of the complementary channel is by definition a positive operator and this condition is satisfied only by the two classes of channels. If we show that the complementary channels of cloning channels $\Cl{\ell-1}$ are entanglement-breaking then the Holevo capacity of cloning channels is additive too. It follows from the fact that  entanglement-breaking channels single-letterize the classical capacity~\cite{entbreak} together with the result of Ref.~\cite{ComplCh} showing that a channel is additive if and only if its complementary channel is additive.
\begin{proof}[Proof of Lemma~\ref{lem:S_ell-entbreaking}]
   Let us first take a look at Fig.~\ref{fig:chnl_struct} representing the intricate mutual dependence of $\Cl{\ell-1}$ for all $\ell\geq3$. By a direct calculation we verify that $\S_2^c$ is entanglement-breaking since $R_2=(\openone\otimes\S_2^c)(\Phi^+)$ is a PPT state which stands for
    positive partial transpose. Therefore, it is a separable state (the input and output Hilbert space is two-dimensional). As a consequence we may write
   \begin{equation}\label{eq:entbreakoutput2}
    R_2=\sum_iq_i\chi_i\otimes\upsilon_i=\sum_iq_i\chi_i
    \otimes\Big( {1\over2}\openone^{(2)} + {1\over3}\sum_{j=1}^3 \tilde n_{j}J_j^{(2)}\Big)_i
   \end{equation}
   where $\chi_i,\upsilon_i$ are positive operators and $0\leq q_i\leq1,\sum_iq_i=1$. The second equation is valid in general since the $su(2)$ algebra generators form an orthogonal basis. To continue let us recall how we determine the output of the rest of complementary channels $\S_{\ell-1}^c$. We found the answer in Eq.~(\ref{StokesplusConjugateStokes2}). We get the output state by a mere exchange of the $J^{(2)}$ generators of the $su(2)$ algebra for higher-dimensional generators $J^{({\ell-1})}$. The coefficients $\tilde n_i$ stay preserved and $\bar\k_{\ell-1}$ is a density operator for all $\ell$. Hence if we write $R_{\ell-1}=(\openone\otimes\S_{\ell-1}^c)(\Phi^+)$ then
   \begin{equation}\label{eq:entbreakoutput_ell}
    R_{\ell-1}={2\over\ell(\ell-1)}\sum_iq_i\chi_i
    \otimes\Big( \openone^{(\ell-1)}/2 + \sum_{j=1}^3 \tilde n_{j}J_j^{(\ell-1)}\Big)_i
   \end{equation}
   is again a valid quantum state and moreover manifestly separable. It follows that all $\S_{\ell-1}^c$ are entanglement-breaking.
\end{proof}
\begin{rem}
    For a generalization of the above Lemma to the case of complementary channels of the qudit cloners see~\cite{capregion_qudits}.
\end{rem}
\begin{figure}[t]
\begin{center}
    \resizebox{8cm}{4.4cm}{\includegraphics{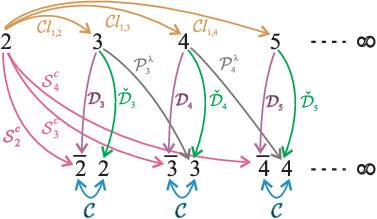}}
    \caption{The relation among various cloning channels $\Cl{\ell-1}$ is sketched here. Numbers indicate the dimension of the particular Hilbert space $\HH$ for the output of $\Cl{\ell-1}$ and the output of their complementary channels $\S_{\ell-1}^c$. In the lower row we further distinguish between the space occupied by the actual output $\bar\k_{\ell-1}$ of the channel $\S_{\ell-1}^c$ (barred numbers) and its complex conjugated version (unbarred numbers). $\P_\ell^\lam$ is a depolarizing channel where $\lam=(\ell-1)/(\ell+1)$ and $\Dcdell$ is a conjugate degrading map.}
    \label{fig:chnl_struct}
    \end{center}
\end{figure}
\begin{proof}[Proof of Theorem~\ref{thm1}]
     All entanglement-breaking channels single-letterize the classical capacity quantity~\cite{entbreak}. Since the  complementary channels of cloning channels $\Cl{\ell-1}$ are entanglement-breaking $\Cl{\ell-1}$ are therefore also additive.
\end{proof}
\begin{cor}
    The previous theorem further shed some light on the properties of $1\to(\ell-1)$ cloning channels. Invoking the result of Cubitt et al.~\cite{degch_study} stating that all entanglement-breaking channels are anti-degradable it follows that all cloning channels $\Cl{\ell-1}$ are also degradable. Degradable channels are known to possess a single-letter formula for the quantum capacity so this result confirms the same findings from Ref.~\cite{ConjDeg} based on the property of conjugate degradability.
\end{cor}

\begin{cor}
     We are now able to explicitly write down the formula for the classical capacity. Since all $\Cl{\ell-1}$ are covariant we suitably choose the coefficients $\a,\b$ such that states $\e_\ell$ from Eq.~(\ref{StokesplusConjugateStokes1}) are diagonal ($\a=1,\b=0$). Then
    $$
    \e_\ell={1\over\Delta_\ell}\sum^{\ell-1}_{k=0}k\kb{k}{k},
    $$
    where $\Delta_\ell=\ell(\ell-1)/2$. Hence, considering $\log{f}=\log{\ell}$ in Eq.~(\ref{equivlocal}), we get
   \begin{equation}\label{eq:cloningchannelclasscap}
    C(\Cl{\ell-1})=1-\log{(\ell-1)}+{1\over\Delta_\ell}\sum_{k=0}^{\ell-1}k\log{k}.
   \end{equation}

\end{cor}
A potentially useful consequence of Lemma~\ref{lem:S_ell-entbreaking} is the fact that a composition of a cloning channel and a depolarizing channel $\P^\lam_{\ell}\circ\Cl{\ell-1}=\Cl{\ell-1}\circ\P^\lam_2$ (holds for $\lam=(\ell-1)/(\ell+1)$) is weakly additive (this is not needed for the purpose of this paper). Recall the definition of the depolarizing channel $\P^\lam_\ell(\r)=\lam\r+(1-\lam)\openone^{(\ell)}/\ell$ for ${-1/(\ell^2-1)}\leq\lam\leq1$ To prove the claim we make use of the following lemma.
\begin{lem}
Let $\N$ be an additive channel (weakly or not). Then if $\M=\C\circ\N$ is another channel, where $\C$ denotes complex conjugation in a given basis, $\M$ is weakly additive as well.
\end{lem}
\begin{rem}
The lemma is intended to hold for situations when the map $\M$ is a CP map. This is not always the case (for instance, if $\N$ is an identity channel).
\end{rem}
\begin{proof}
We write
\begin{subequations}
    \begin{align}
        &C_{\rm Hol}((\C\circ\N)^{\otimes n})\\
        &=\sup_{\{p_i\r_i\}}{\bigg\{S\bigg(\sum_ip_i\overline{\N^{\otimes n}(\r_i)}\bigg)-\sum_i p_iS\bigg(\overline{\N^{\otimes n}(\r_i)}\bigg)\bigg\}}\\
        &=\sup_{\{p_i\r_i\}}{\bigg\{S\bigg(\sum_ip_i{\N^{\otimes n}(\r_i)}\bigg)-\sum_i p_iS\bigg({\N^{\otimes n}(\r_i)}\bigg)\bigg\}}\\
        &=C_{\rm Hol}(\N^{\otimes n}).
    \end{align}
\end{subequations}
The first eqaulity follows from definition (Eq.~(\ref{HolevoCap})) and the second equality holds thanks to $S(\s)=S(\bar\s)$ which is valid for all $\s$. Since we assume that
    \begin{equation}
        C_{\rm Hol}(\N^{\otimes n})=nC_{\rm Hol}(\N),
    \end{equation}
then due to $nC_{\rm Hol}(\N)=nC_{\rm Hol}(\C\circ\N)$ we finally get the lemma statement
\begin{equation}
    C_{\rm Hol}((\C\circ\N)^{\otimes n})=nC_{\rm Hol}(\C\circ\N).
\end{equation}
\end{proof}
Looking at Fig.~\ref{fig:chnl_struct} we can see why additivity of $\S_{\ell}^c$ for $\ell\geq3$ implies additivity of $\P^\lam_{\ell}\circ\Cl{\ell-1}$. The reason lies in the fact that $\P^\lam_{\ell}\circ\Cl{\ell-1}$ composed with complex conjugation is equal to $\S_{\ell}^c$ for all $\ell$.

\subsection{The classical capacity of the Unruh channel}

The output of the Unruh channels is a weighted direct sum of outputs of cloning channels $\Cl{\ell-1}$ for all $\ell$. From Theorem~\ref{thm1} we  know that the Holevo capacity of all of them is additive. This directly leads to the proof of additivity of the Holevo capacity for the Unruh channels itself.
\begin{thm}\label{thm2}
    The infinite-dimensional Unruh channel studied in~\cite{UnruhCh} is additive
\end{thm}
First, let us present a lemma.
\begin{lem}\label{lem:Smin_directsum}
    Let $\A,\B$ be additive and covariant but otherwise arbitrary finite-dimensional channels whose input Hilbert spaces are of the same dimension. Then a channel $\G:\FF\big(\HH\big)\to\FF\big(\HH_\A\oplus\HH_\B\big)$ is additive for any ensemble $\{q_\A,q_\B\}$.
\end{lem}
\begin{proof}
    The channel output is unitarily equivalent to
    \begin{equation}\label{channleaction}
        \r\xrightarrow{\G}q_\A\r_\A\oplus q_B\r_B\equiv\kb{0}{0}\otimes q_\A\r_\A+\kb{1}{1}\otimes q_\B\r_\B.
    \end{equation}
    Defining $\T$ to be an arbitrary channel we see that for any input pure state $\o$ of the channel $\G\otimes\T$ the output state is a block-diagonal matrix $\s=q_\A(\A\otimes\T)(\o)\oplus q_\B(\B\otimes\T)(\o)$. Thus, $S(\s)=S(\{q_\A,q_\B\})+q_\A S((\A\otimes\T)(\o))+q_\B S((\B\otimes\T)(\o))$.
    Hence
    \begin{align}
    S^{\rm min}(\G\otimes\T)&=S(\{q_\A,q_\B\})+q_\A\min_\o\{S((\A\otimes\T)(\o))\}\nonumber\\
    &\ \ \ +q_\B\min_{\o^\p}\{S((\B\otimes\T)(\o^\p))\}\nonumber\\
    &=S(\{q_\A,q_\B\})\nonumber\\
    &\ \ \ +q_\A S(\A(\varphi))+q_\B S(\B(\varphi))+\Smin{\T}\nonumber\\
    &\equiv S(\G(\varphi))+\Smin{\T}
    \end{align}
    using the properties of $\A$ and $\B$.
\end{proof}
\begin{proof}[Proof of Theorem~\ref{thm2}]
The proof is a direct application of the previous lemma since the Unruh channel happens to be $\U(\varphi)=\bigoplus_{\ell=2}^\infty p_\ell\Cl{\ell-1}(\varphi)$ where $p_\ell=(1-z)^3z^{\ell-2}(\ell-1)\ell/2,0\leq z<1$. The channel $\U(\varphi)$ is the same channel as in Eq.~(\ref{squeeze}) where it is written in the Fock space representation.

We show using the block-diagonal structure of the output state and the unitary covariance of the Unruh channel that the inductive process described above approximates the channel output with an arbitrary precision for any input qubit.
Namely, let us denote a partial sum $c_K=\sum_{\ell=2}^Kp_\ell$. We get
\begin{equation}
    c_K={1\over2}\big(2+2(K^2-1)z^K-K(K+1)z^{K-1}-K(K-1)z^{K+1}\big)
\end{equation}
and so $\lim_{K\to\infty}{c_K}=1$ for all $0\leq z<1$.
\end{proof}
\begin{rem}
    Note that the channel input for otherwise infinite-dimensional Unruh channel is naturally energy constrained since the set of input states is limited to qubits.
\end{rem}
\begin{cor}
    Theorem~\ref{thm2} enables us to bring up the formula for the classical capacity of the Unruh channel. Using Eq.~(\ref{HolevoCap}) and the covariance of the Unruh channel we get
     \begin{equation}\label{capUnruh}
        C(\U)=1-\sum_{\ell=2}^\infty p_\ell\log{(\ell-1)} + \sum_{\ell=2}^\infty {p_\ell\over\Delta_\ell}\sum_{k=0}^{\ell-1} k\log{k}.
    \end{equation}
    The plot in Fig.~\ref{fig:chnl_UnClCap} depicts the Holevo capacity as a function of the parameter $z$.
\end{cor}
\begin{figure}[t]
\begin{center}
    \resizebox{9.4cm}{6cm}{\includegraphics{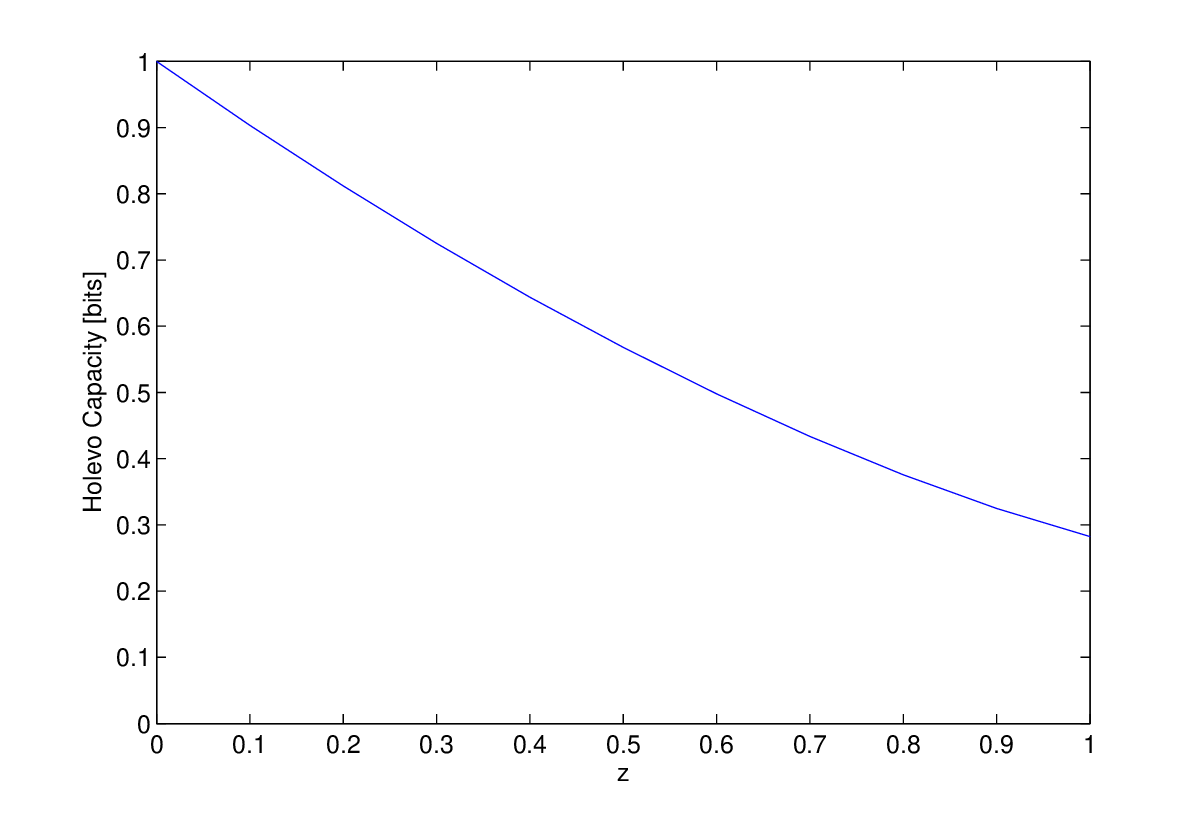}}
    \caption{The Holevo capacity for the Unruh channel as a function of $z$. Recall that $z$ itself is  a function of the proper acceleration of a non-inertial observer. Note that in the limit of infinite acceleration ($z\to1$) the capacity converges to a non-zero value. This value is the same as the capacity for $\Cl{\ell-1}$ for $\ell\to\infty$ from Eq.~(\ref{eq:cloningchannelclasscap}) since the peak of probability distribution is `moving' towards infinity with growing acceleration.}
    \label{fig:chnl_UnClCap}
    \end{center}
\end{figure}

\section{Degrading map construction}\label{sec:degrad}

Let us attempt to construct degrading maps for  several  low-dimensional cloning channels $\Cl{\ell-1}$ we studied in the previous section.

We first analyze the case $\ell=3$. Looking at Eqs.~(\ref{StokesplusConjugateStokes1}) and (\ref{StokesplusConjugateStokes2}) we see that the complementary output of every $\Cl{\ell-1}$ is effectively conjugated with respect to the channel output. This fact together with the unitary covariance of cloning channels leads to the condition similar to Eq.~(\ref{covcond})
\begin{equation}\label{channelcovariance}
    \overline{\Dthree(r_1(g)\varrho\,r_1(g)^\dagger)}=r_2\,\overline{\Dthree(\varrho)}\,r_2^\dagger,
\end{equation}
where the presence of bars is the result of complex conjugation.  In this case, $r_2$ and $r_1$ is the two- and three-dimensional irrep of $g\in G=SU(2)$, respectively. This is, however, the same as the contravariance condition
\begin{equation}\label{channelcontracovariance}
    \Dthree(r_1\varrho\,r_1^\dagger)=\overline{r_2(g)}\,\Dthree(\varrho)\,r_2(g)^T.
\end{equation}
By rephrasing this condition within the Choi-Jamio\l kowski isomorphism~\cite{jamiolk} we get
\begin{equation}\label{jamicovariance}
    \big[\,\overline{R_\Dthree},r_2\otimes r_1\big]=0
\end{equation}
when $\overline{R_\Dthree}$ is a positive semidefinite matrix corresponding to the CP map $\Dthree$. One of Schur's lemmas dictates $\overline{R_\Dthree}=\bigoplus_ic_i\Pi_i$ where $c_i\geq0$ and $\Pi_i$ are projectors into the subspaces of the split product $[2]\otimes[3]=[2]\oplus[4]$ and thus $\overline{R_\Dthree}\equiv R_\Dthree$. We insert $R_\Dthree$ into $\Dthree(\e_{3_{in}})=\Tr{in}\big[\left(\openone_{out}\otimes\bar\e_{3_{in}}\right)R_\Dthree\big]$ since we are looking for such $R_\Dthree$ that $\bar\k_2=\Dthree(\e_{3})$ where (index $in$ omitted)
\begin{gather}\label{inputoutput_ell3}
   \e_3={1\over3}
    \begin{pmatrix}
                               2|\a|^2 & \sqrt{2}\a\bar \b & 0 \\
                               \sqrt{2}\bar \a\b & 1 & \sqrt{2}\a\bar \b \\
                               0 & \sqrt{2}\bar \a\b & 2|\b|^2 \\
    \end{pmatrix},\\
   \bar\k_2={1\over3}
    \begin{pmatrix}
                               |\a|^2+1 & \bar \a\b \\
                               \a\bar \b & |\b|^2+1\\
    \end{pmatrix}.
\end{gather}
In other words, we maximize the fidelity between these two states checking whether it reaches one for some $c_1,c_2$ considering the constraints $c_{1,2}\geq0$ and $\Tr{out}\big[R_\Dthree\big]=\openone^{(3)}$. Because we are dealing with mixed states, we use the fidelity expression due to Bures which simplifies for two-dimensional matrices~\cite{FidelityDim2} as
\begin{equation}\label{Buresfid}
    F(\Dthree(\e_{3}),\bar\k_2)=\Tr{}[\Dthree(\e_{3})\bar\k_2]+2\sqrt{\Det[\Dthree(\e_{3})]\Det[\bar\k_2]}.
\end{equation}
As expected from the results in section~\ref{subsec:infseq}, the fidelity reaches one. In general, the decomposition $\overline{R_\Dell}=\bigoplus_ic_i\Pi_i$ might be difficult to determine. Nevertheless, the good news is that an ansatz can be made. Following the lowest-dimensional exact solutions for the form of the degrading maps of $\Cl{2}$ and $\Cl{3}$ we observe that the only surviving coefficient $c_i$ from the expression for the Jamio\l kowski matrices is the one accompanying the highest irrep of the $SU(2)$ tensor product. Indeed, applying this guess on a few more $1\to (\ell-1)$ cloning channels ($\ell=5,6,7$) it always yields the sought degrading map. So we know that $\Cl{\ell-1}$ are degradable and the construction of the degrading maps might be hard for large $\ell$ but verification of the ansatz is very fast even for large $\ell$.

\section{Conclusions}
The general non-additivity result for the classical capacity of quantum channels is in some sense very satisfactory. Not only did entanglement prove to be useful for the transmission of classical information but it will spark even more effort to find out what makes a channel (non-)additive. Also, some novel strategies may be found to prove (non-)additivity for particular channels as it is now known that there is no general proof.  In this paper we investigated an infinite family of channels we call $1\to (\ell-1)$ cloning channels ($\ell=2\dots\infty$) which are the incarnations of universal quantum cloning machines for qubits.
To prove additivity of the Holevo capacity for  cloning channels we used the fact that cloning channels enjoy the property of being conjugate degradable channels. Conjugate degradable channels already prove to be a useful concept since it is known that their quantum capacity has a single-letter formula. We have therefore found an infinite family of channels for which both the classical and quantum capacity is easily calculable. Also, we were able to prove that $1\to (\ell-1)$ cloning channels are degradable. Furthermore, using the fact that cloning channels are intimately related to an infinite-dimensional channel called the Unruh channel  we were also able to present  the additivity proof of the Unruh channel which otherwise seems intractable. The infinite-dimensional Unruh channel is now a member of a rare family of channels with both capacities easily calculable and non-zero since the existence of a single-letter quantum capacity formula has been proved elsewhere. This result might find an important future application in quantum field theory in curved spacetime considering the prominent role the Unruh channel plays in this branch of modern physics.

\section*{acknowledgments}
The author is grateful to Patrick Hayden for comments and discussions. The comments on early versions of the manuscript made by Dave Touchette, Min-Hsiu Hsieh and Mark Wilde are also appreciated. The work was supported  by QuantumWorks and by a grant from the Office of Naval Research  (N000140811249).

\bibliographystyle{IEEEtran}

\end{document}